\documentclass[11pt]{article}
\usepackage{amsmath,amsfonts,amsthm,amssymb,color}
\usepackage{fancybox}
\usepackage{graphics}
\usepackage{graphicx}
\usepackage{bbm}

\usepackage[margin=1in]{geometry}

\newtheorem{theorem}{Theorem}[section]

\newtheorem{claim}[theorem]{Claim}

\theoremstyle{definition}

\newenvironment{fminipage}%
  {\begin{Sbox}\begin{minipage}}%
  {\end{minipage}\end{Sbox}\fbox{\TheSbox}}

\def\defeq{\stackrel{\mathrm{def}}{=}}

\def\abs#1{\left|#1  \right|}

\def\norm#1{\left\| #1 \right\|}

\def\calE{\mathcal{E}}

\def\calS{\mathcal{S}}

\def\eps{\epsilon}

\newcommand\PPi{\boldsymbol{\Pi}}

\newcommand\uu{\boldsymbol{\mathit{u}}}
\newcommand\vv{\boldsymbol{\mathit{v}}}
\newcommand\ww{\boldsymbol{\mathit{w}}}
\newcommand\yy{\boldsymbol{\mathit{y}}}
\newcommand\zz{\boldsymbol{\mathit{z}}}
\newcommand\xx{\boldsymbol{\mathit{x}}}

\renewcommand\AA{\boldsymbol{\mathit{A}}}

\newcommand\II{\boldsymbol{\mathit{I}}}
\newcommand\JJ{\boldsymbol{\mathit{J}}}

\newcommand\MM{\boldsymbol{\mathit{M}}}

\newcommand{\one}{\mathbf{1}}

\def\tr{\text{tr}}

\def\cov{\text{Cov}}
\def\ssplit{\text{Set-Splitting}}

\def\dims{d}
\def\n{n}
\def\N{N}

\title{Hardness Results for Minimizing the Covariance of Randomly Signed Sum of Vectors}
\author{Peng Zhang \\ pz149@cs.rutgers.edu \\ Department of Computer Science \\ Rutgers University}

\begin{document}

\maketitle

\begin{abstract}
  Given vectors $\vv_1, \ldots, \vv_n \in \mathbb{R}^d$ with Euclidean norm at most $1$ and $\xx_0 \in [-1,1]^n$, 
  our goal is to sample a random signing $\xx \in \{\pm 1\}^n$ with
  $\mathbb{E}[\xx] = \xx_0$ such that the operator norm of the covariance of the signed sum of the vectors $\sum_{i=1}^n \xx(i) \vv_i$
  is as small as possible. This problem arises from the algorithmic discrepancy theory and 
  its application in the design of randomized experiments. 
  It is known that one can sample a random signing with expectation $\xx_0$
  and the covariance operator norm at most $1$.
  
  In this paper, we prove two hardness results for this problem.
  First, we show it is NP-hard to distinguish a list of vectors 
  for which there exists a random signing with expectation ${\bf 0}$
  such that the operator norm is $0$ from those for which any signing with expectation ${\bf 0}$
  must have the operator norm $\Omega(1)$.
  Second, we consider $\xx_0 \in [-1,1]^n$ whose entries are all around an arbitrarily fixed $p \in [-1,1]$.
  We show it is NP-hard to distinguish a list of vectors for which 
  there exists a random signing with expectation $\xx_0$
  such that the operator norm is $0$ from those for which any signing with expectation ${\bf 0}$
  must have the operator norm $\Omega((1-\abs{p})^2)$.
\end{abstract}

\section{Introduction}

Given a list of $\n$ vectors $\mathcal{V} = \vv_1, \ldots, \vv_{\n} \in \mathbb{R}^{\dims}$
and a vector $\xx_0 \in [-1,1]^{\n}$, our goal is to sample a random signing vector $\xx \in \{\pm 1\}^{\n}$
with $\mathbb{E}[\xx] = \xx_0$ such that 
the covariance of the signed sum of the vectors, 
$$
\cov(\mathcal{V}, \xx) \defeq
\cov\left( \sum_{i=1}^{\n} \xx(i) \vv_i \right)
= \mathbb{E} \left[ \left(
  \sum_{i=1}^{\n} (\xx(i) - \xx_0(i)) \vv_i
\right) \left(
  \sum_{i=1}^{\n} (\xx(i) - \xx_0(i)) \vv_i
\right)^\top \right],
$$
has the minimum operator norm. 
Here, $\xx(i)$ is the $i$th entry of $\xx$.
Since the covariance scales quadratically with the maximum Euclidean norm 
of vectors $\vv_1, \ldots, \vv_n$,
without loss of generality, we assume all $\vv_1, \ldots, \vv_{\n}$ have Euclidean norms at most $1$.

This problem arises from the algorithmic discrepancy theory and 
its application in the design of randomized experiments. 
A stronger version of the problem was first studied by Dadush, Garg, Lovett, and Nikolov \cite{DGLN19},
aiming to provide an algorithmic proof of Banaszczyk's discrepancy problem \cite{banaszczyk98}.
Here, the goal is to sample $\xx \in \{\pm 1\}^n$ with $\mathbb{E}[\xx] = \xx_0$ such that 
$\sum_{i=1}^n (\xx(i) - \xx_0(i)) \vv_i$ is $\sigma$-subgaussian.
The $\sigma$-subgaussianity immediately implies the operator norm of $\cov(\mathcal{V}, \xx)$,
denoted by $\norm{\cov(\mathcal{V}, \xx)}$, is at most $\sigma^2$.
Bansal, Dadush, Garg, and Lovett \cite{BDGL18} designed a polynomial time algorithm, 
called the Gram-Schmidt Walk, that outputs a random $\xx \in \{\pm 1\}^{\n}$ 
achieving $\sigma \le \sqrt{40}$. 
This upper bound was then improved to $\sigma \le 1$, by Harshaw, S{\"a}vje, Spielman, and Zhang \cite{HSSZ19},
which is tight and the equality holds when $\n = \dims$ and 
$\vv_1, \ldots, \vv_{\n}$ are the $\n$ standard basis vectors in $\mathbb{R}^{\n}$.
Building on the Gram-Schmidt Walk algorithm,
Harshaw, S{\"a}vje, Spielman, and Zhang \cite{HSSZ19} proposed the Gram-Schmidt Walk Design
to balance covariates in randomized experiments widely used in causal inference.

All the above upper bounds for $\norm{\cov(\mathcal{V}, \xx)}$, achieved by the Gram-Schmidt Walk, 
are independent of the optimal value, denoted by 
\[
  C(\mathcal{V}, \xx_0) \defeq \min_{\xx \in \{\pm 1\}^n: \mathbb{E}[\xx] = \xx_0} 
\norm{\cov(\mathcal{V}, \xx)}.
\] 
It is natural to ask whether we can efficiently sample a random $\xx \in \{\pm 1\}^n$ with $\mathbb{E}[\xx] = \xx_0$
such that $\norm{\cov(\mathcal{V}, \xx)}$ (approximately) equals the optimal value?
In this paper, we prove strong hardness results for this question.

\begin{theorem}
  There exists a constant $C_1 > 0$ such that given a list of vectors $\mathcal{V}$ 
  of Euclidean norm $1$, it is NP-hard to distinguish whether $C(\mathcal{V}, {\bf 0}) = 0$
  or $C(\mathcal{V}, {\bf 0}) > C_1$.
  \label{thm:main1}
\end{theorem}

Theorem \ref{thm:main1} concerns $\xx_0 = {\bf 0}$; that is, for every $i \in \{1,\ldots,n\}$, 
the marginal probability of $\xx(i)$ being $1$ or $-1$ equals $1/2$.
One may expect that the hardness comes from the balanced marginal probability for each $\xx(i)$.
If the marginal probability of each $\xx(i)$ changes towards $0$ or $1$, 
 the problem may become more tractable. 
In particular, when $\xx_0 \in \{\pm 1\}^n$, it is clear that $C(\mathcal{V}, \xx_0) = 0$ 
due to no randomness. 
Our Theorem \ref{thm:main2} concerns $\xx_0 \neq {\bf 0}$.
It shows a gap, parameterized only by entries of $\xx_0$, between the covariance operator norms 
of two cases which are NP-hard to distinguish.

\begin{theorem}
  There exists a constant $C_2 > 0$ such that the following holds:
  For any $p,q \in [-1,1]$, there exists $\xx_0 \in \{p, p+(1-\abs{p})q, p-(1-\abs{p})q \}^{n}$
  such that given a list of $n$ vectors $\mathcal{V}$ of Euclidean norm at most $1$, it is NP-hard to distinguish whether 
  $C(\mathcal{V}, \xx_0) = 0$ or $C(\mathcal{V}, \xx_0) > C_2 (1-\abs{p})^2 q^2$.
  \label{thm:main2}
\end{theorem}

The parameters $p,q$ in Theorem \ref{thm:main2} may depend on $d$ or $n$.
When the parameter $q$ is a constant near $0$, all the entries of $\xx_0$ are near $p$;
Theorem \ref{thm:main2} implies that it is NP-hard to distinguish whether $C(\mathcal{V}, \xx_0) = 0$
or $C(\mathcal{V}, \xx_0) = \Omega((1-\abs{p})^2)$.
When $\abs{p}$ increases, the gap between the two cases in Theorem \ref{thm:main2} decreases.
In particular, when $\abs{p}$ goes to $1$, the gap goes to $0$.

\paragraph*{Proof ideas.}
Our proofs of Theorem \ref{thm:main1} and \ref{thm:main2} build on reductions from the 2-2 \ssplit \ problem, 
for which Guruswami \cite{guruswami04} proved strong NP-hardness results.
Roughly speaking, in the 2-2 \ssplit \ problem, we are given a universe $U$ 
and a family $\calS$ of $4$-subsets of $U$, and our goal is to assign each element in the universe $1$ or $-1$ 
to maximize the number of ``split'' sets in $\calS$ (a split set has half elements assigned $1$ and half $-1$).
The 2-2 \ssplit \ problem is closely related to the problem of signing vectors to minimize their discrepancy.
Reducing from the 2-2 \ssplit \ problem,
Charikar, Newman, and Nikolov \cite{CNN11} proved NP-hardness results for 
minimizing Spencer's discrepancy \cite{spencer85,bansal10}; 
Spielman and Zhang \cite{SZ22} proved NP-hardness results for minimizing 
Weaver's discrepancy \cite{weaver04,MSS15,BCMS19}.
Our proofs are inspired by those from \cite{CNN11} and \cite{SZ22}. 
However, our constructions are different from these two papers due to the different 
notions of discrepancy.
The proof of Theorem \ref{thm:main1} is a direct reduction from 
the 2-2 \ssplit \ problem, together with an inequality between 
matrix operator norm and matrix trace.
The proof of Theorem \ref{thm:main2} is slightly more involved due to the requirement of the nonzero expectation of $\xx$.
Comparing to the proof of Theorem \ref{thm:main1},
we introduce auxiliary input vectors so that we can construct a signing $\xx$ with the required expectation 
and covariance ${\bf 0}$ whenever such a signing exists,
and we employ an orthogonal projection matrix to force the sum of signed auxiliary vectors 
is almost zero with a sufficiently large probability
under which the signed sum of all the input vectors 
behave similarly to those in Theorem \ref{thm:main1} (without auxiliary vectors).

\paragraph*{Organization of the rest of the paper.}
In Section \ref{sec:notation}, we introduce some preliminaries and notations, 
and formally define the 2-2 \ssplit \ problem and its variants, and state 
the known hardness results.
We prove Theorem \ref{thm:main1} in Section \ref{sec:proofs1} and Theorem \ref{thm:main2} in Section \ref{sec:proof2}.

\section{Preliminaries and Notations}
\label{sec:notation}

\subsection{Matrices and Vectors}

Given a vector $\xx \in \mathbb{R}^n$, we let $\xx(i)$ be the $i$th entry of $\xx$.
Given a matrix $\AA \in \mathbb{R}^{m \times n}$, we let $\AA(i,j)$ be the $(i,j)$th entry of $\AA$.
The Euclidean norm of $\xx$ is $\norm{\xx} \defeq \sqrt{\sum_{i=1}^n \xx(i)^2 }$. 
The \emph{operator norm} of $\AA$ is 
\[
\norm{\AA} \defeq \sup_{\xx \in \mathbb{R}^n} \frac{\norm{\AA \xx}}{\norm{\xx}}.  
\]
When $\AA$ is a square matrix, the trace of $\AA$ is the sum of the entries on its main diagonal, denoted by $\tr(\AA)$.
The trace of $\AA$ equals the sum of the eigenvalues of $\AA$.
In addition, we will use $\one_n$ for the all-$1$ vector in $n$ dimensions 
and ${\bf 0}_{n}$ for the all-$0$ vector, 
and use $\JJ_{m \times n}$ for the all-$1$ matrix in $m \times n$ dimensions
and ${\bf 0}_{m \times n}$ for the all-$0$ matrix .
When the context is clear, we drop the subscription for dimensions.
We will use $\II$ for the identity matrix.

\subsection{2-2 Set-Splitting Problem} 
\label{sec:setsplit}

Our proofs of Theorem \ref{thm:main1} and \ref{thm:main2} build on reductions from the 2-2 \ssplit problem.
In the 2-2 \ssplit \  problem, we are given a universe $U = \{1,2,\ldots,n\}$ and a family of sets 
$\calS = \{S_1, \ldots, S_m\}$ in which each $S_j$ consists of $4$ distinct elements from $U$.
Our goal is to find an assignment of the $n$ elements in $U$, denoted by $\zz \in \{\pm 1\}^n$,
to maximize the number of sets in $\calS$ in which the values of its elements sum up to $0$.
We say an assignment $\zz$ \emph{2-2-splits} (or simply, splits) a set $S_j \in \calS$ 
if $\sum_{i \in S_j} \zz(i) = 0$; we say $\zz$ \emph{unsplits} $S_j$ if $\sum_{i \in S_j} \zz(i) \in \{\pm 2, \pm 4\}$.
We say an instance of the 2-2 \ssplit \  problem is \emph{satisfiable} if 
there exists an assignment that splits all the sets in $\mathcal{S}$.
We say an instance is \emph{$\gamma$-unsatisfiable} if any assignment 
must unsplit at least $\gamma$ fraction of the sets in $\calS$.
Given a number $b \ge 1$, a 2-2 \ssplit \ instance is called a $(b,2$-$2)$ \ssplit \ instance if each element in $U$ appears in at most 
$b$ sets in $\calS$.
In a $(b,2$-$2)$ \ssplit \ instance, we have $4m \le bn$.

\begin{theorem}[\cite{guruswami04}]
  For any constant $\eps > 0$, there exists a constant $b$ such that it is NP-hard to 
  distinguish satisfiable $(b,2$-$2)$ \ssplit \  instances from $(1/12 - \eps)$-unsatisfiable instances.  
  \label{thm:ssplit_hard}
\end{theorem}

A similar hardness result holds for $b = 3$. We will need it 
for our constructions.

\begin{theorem}[\cite{SZ22}]
  There exists a constant $\gamma > 0$ such that it is NP-hard to distinguish 
  satisfiable instances of the $(3,2$-$2)$ \ssplit \  problem from $\gamma$-unsatisfiable instances.
  \label{thm:322sspit_hard}
\end{theorem}

\section{Proof of Theorem \ref{thm:main1}}
\label{sec:proofs1}

In this section, we prove Theorem \ref{thm:main1}.

Given a $(3,2$-$2)$ \ssplit \  instance where $\abs{U} = n$ and $\abs{\calS} = m$,
we will construct a list of $\N$ vectors $\mathcal{V} = \vv_1, \ldots, \vv_{\N} \in \mathbb{R}^\dims$ each of Euclidean norm $1$ such that 
(1) if the given $(3,2$-$2)$ \ssplit \  instance is satisfiable, then $C(\mathcal{V}, {\bf 0}) = 0$,
and (2) if the given $(3,2$-$2)$ \ssplit \  instance is $\gamma$-unsatisfiable, then 
$C(\mathcal{V}, {\bf 0}) > C_1$.

For each element $i \in U$, let $A_i \subset \{1,\ldots,m\}$ consist of the 
indices of the sets that contain $i$. 
For each element $i$ that appears in exactly $1$ set in $\calS$ (that is, $\abs{A_i} = 1$), we create $4$ new sets and $2$ new elements.
For each element $i$ that appears in $2$ sets in $\calS$, we create $5$ new sets and 
$3$ new elements. Let $B_i$ be the set consisting of the indices of the newly created sets for element $i$.
Suppose there are $n_1$ elements in $U$ that appear in exactly $1$ set in $\calS$
and $n_2$ elements that appear in $2$ sets.
We set 
$$\dims = m + 4n_1 + 5n_2 \le m + 5n 
\text{ and } N = n + 2n_1 + 3n_2 \le 4n.$$
Consider each element $i \in U$. 
There are $3$ cases depending on how many sets in $\calS$ containing $i$:
\begin{enumerate}
  \item Element $i$ appears in $3$ sets in $\calS$: We define $\vv_i \in \mathbb{R}^d$ 
  such that $\vv_i(j) = \frac{1}{\sqrt{3}}$ for $j \in A_i$ and $\vv_i(j) = 0$ otherwise.

  \item Element $i$ appears in $1$ set in $\calS$: Suppose $B_i = \{i_1, i_2, i_3, i_4\}$.
  We define $\vv_i \in \mathbb{R}^d$ such that $\vv_i(j) = \frac{1}{\sqrt{3}}$ for $j \in A_i \cup \{i_1, i_2\}$
  and $\vv_i(j) = 0$ otherwise. We define two more vectors: (1) $\uu_{i,1} \in \mathbb{R}^d$ such that 
  $\uu_{i,1}(i_1) = \uu_{i,1}(i_3) = \uu_{i,1}(i_4) = \frac{1}{\sqrt{3}}$ and 
  $\uu_{i,1}(j) = 0$ for all other $j$'s, and (2) $\uu_{i,2} \in \mathbb{R}^d$ such that 
  $\uu_{i,2}(i_2) = -\frac{1}{\sqrt{3}}$ and $\uu_{i,2}(i_3) = \uu_{i,2}(i_4) = \frac{1}{\sqrt{3}}$ and
  $\uu_{i,1}(j) = 0$ for all other $j$'s.

  \item Element $i$ appears in $2$ sets in $\calS$: Suppose $B_i = \{i_1, i_2, i_3, i_4, i_5\}$.
  We define $\vv_i \in \mathbb{R}^d$ such that $\vv_i(j) = \frac{1}{\sqrt{3}}$ for $j \in A_i \cup \{i_1\}$
  and $\vv_i(j) = 0$ otherwise. We define three more vectors (1) 
  $\uu_{i,1} \in \mathbb{R}^d$ such that $\uu_{i,1}(i_1) = \uu_{i,1}(i_2) = \uu_{i,1}(i_3) = \frac{1}{\sqrt{3}}$ and
  $\uu_{i,1}(j) = 0$ for all other $j$'s, (2) $\uu_{i,2} \in \mathbb{R}^d$ such that 
  $\uu_{i,2}(i_2) = \uu_{i,2}(i_4) = \uu_{i,2}(i_5) = \frac{1}{\sqrt{3}}$ and 
  $\uu_{i,2}(j) = 0$ for all other $j$'s, and (3) $\uu_{i,3} \in \mathbb{R}^d$ such that 
  $\uu_{i,3}(i_3) = -\frac{1}{\sqrt{3}},  \uu_{i,3}(i_4) = \uu_{i,3}(i_5) = \frac{1}{\sqrt{3}}$ and 
  $\uu_{i,3}(j) = 0$ for all other $j$'s.
\end{enumerate}
We let $\vv_{n+1}, \ldots, \vv_{N}$ be the vectors $\uu_{i,h}$'s constructed above.
We can check that all $\vv_1, \ldots, \vv_N$ have Euclidean norm $1$.

\begin{claim}
  
  For any $\zz \in \{\pm 1\}^n$, we can construct an $\yy \in \{\pm 1\}^N$
  such that (1) $\yy(i) = \zz(i)$ for $i \in \{1,\ldots,n\}$ and 
  (2) $\sum_{i=1}^N \yy(i) \vv_i (j) = 0$ for all $j \in \{m+1, \ldots, d\}$.
\label{clm:sign_u}  
\end{claim}

\begin{proof}
  We only need to determine the signs for the vectors $\uu_{i,h}$'s constructed for elements 
  appearing in less than 3 sets in $\calS$ and to check the coordinates of $\sum_{i=1}^N \yy(i) \vv_i$
  whose indices in $B_i$'s. 
  Let $i \in U$ be an element that appears in $1$ set in $\calS$. 
  The subvectors of $\vv_i, \uu_{i,1}, \uu_{i,2}$ restricted to the coordinates 
  in $B_i$ are:
  \[
  \begin{pmatrix}
    1 \\
    1 \\
    0 \\
    0
  \end{pmatrix},
  \begin{pmatrix}
    1 \\
    0 \\
    1 \\
    1
  \end{pmatrix},
  \text{ and }  
  \begin{pmatrix}
    0 \\
    -1 \\
    1 \\
    1
  \end{pmatrix}.
  \]
  We choose the signs  in $\yy$ for $\uu_{i,1}, \uu_{i,2}$ to be $-\zz(i)$ and $\zz(i)$, respectively,
  which guarantees the signed sum of the $\vv_1, \uu_{i,1}, \uu_{i,2}$ is ${\bf 0}$
  when restricted to $B_i$. Since any other vector has $0$ for the coordinates in $B_i$,
  we have $\sum_{i=1}^N \yy(i) \vv_i(j) = 0$ for $j \in B_i$. 
  Now, let $i \in U$ be an element that appears in $2$ set in $\calS$. 
  The subvectors of $\vv_i, \uu_{i,1}, \uu_{i,2}, \uu_{i,3}$ restricted to the coordinates 
  in $B_i$ are:
  \[
  \begin{pmatrix}
    1 \\
    0 \\
    0 \\
    0 \\
    0
  \end{pmatrix},
  \begin{pmatrix}
    1 \\
    1 \\
    1 \\
    0 \\
    0
  \end{pmatrix},
  \begin{pmatrix}
    0 \\
    1 \\
    0 \\
    1 \\
    1
  \end{pmatrix},
  \text{ and }
  \begin{pmatrix}
    0 \\
    0 \\
    -1 \\
    1 \\
    1
  \end{pmatrix}.
  \]
  We choose the signs in $\yy$ for $\uu_{i,1}, \uu_{i,2}, \uu_{i,3}$ to be 
  $-\zz(i), \zz(i), -\zz(i)$, respectively.
  This guarantees $\sum_{i=1}^N \yy(i) \vv_i(j) = 0$ for $j \in B_i$. 
  Thus, the constructed $\yy$ satisfies the conditions.
\end{proof}

Suppose the given $(3,2$-$2)$ \ssplit \  instance is satisfiable, meaning there 
exists an assignment $\zz \in \{\pm 1\}^{n}$ such that $\sum_{i=1}^n \zz(i) \vv_i = {\bf 0}$.
We construct a vector $\yy \in \{\pm 1\}^{N}$ as in Claim \ref{clm:sign_u}.
Thus, $\sum_{i=1}^N \yy(i) \vv_i = {\bf 0}$.
We define a random vector $\xx \in \{\pm 1\}^{N}$ such that 
$\xx = \yy$ with probability $1/2$ and $\xx = -\yy$ with probability $1/2$.
Then, $\mathbb{E}[\xx] = {\bf 0}$ and $\cov(\mathcal{V}, \xx) = {\bf 0}$, 
that is, $C(\mathcal{V}, {\bf 0}) = 0$.

Suppose the given $(3,2$-$2)$ \ssplit \  instance is $\gamma$-unsatisfiable, 
meaning that for any assignment $\zz \in \{\pm 1\}^{n}$, at least $\gamma$
fraction of the entries of $\sum_{i=1}^n \zz(i) \vv_i$ are in $\{\pm 2, \pm 4\}$.
Then, for any $\yy \in \{\pm 1\}^N$, at least 
$$\frac{\gamma n}{N} \ge \frac{\gamma}{4}$$
fraction of the entries of $\sum_{i=1}^N \yy(i) \vv_i$ are in $\{\pm 2, \pm 4\}$.
Then, for any random $\xx \in \{\pm 1\}^{N}$ with $\mathbb{E}[\xx] = {\bf 0}$, 
let $\ww = \sum_{i=1}^N \xx(i) \vv_i$,
\begin{align*}
  \norm{\cov \left( \ww  \right)}
  & = \norm{\mathbb{E} \left[ \ww \ww^\top
  \right]}
  \ge \frac{1}{d} \tr \left( \mathbb{E} \left[ \ww \ww^\top
  \right] \right)
  = \frac{1}{d} \mathbb{E} \left[
    \tr(\ww \ww^\top)
  \right] \\
  & 
  \ge \frac{1}{d} \cdot 4 \cdot \frac{\gamma N}{4}
  = \frac{\gamma N}{d}
  \ge \frac{4 \gamma}{23}.
\end{align*}
The last inequality holds since $d \le m+5n \le \frac{23n}{4} \le \frac{23N}{4}$.
That is, $C(\mathcal{V}, {\bf 0}) > \frac{4\gamma}{23}$.
If we can distinguish whether $C(\mathcal{V}, {\bf 0}) = 0$ or $C(\mathcal{V}, {\bf 0}) > \frac{4\gamma}{23}$, 
then we can distinguish whether a $(3,2$-$2)$ \ssplit \ instance is  satisfiable
or $\gamma$-unsatisfiable, which is NP-hard by Theorem \ref{thm:322sspit_hard}.
This completes the proof of Theorem \ref{thm:main1}.

\section{Proof of Theorem \ref{thm:main2}}
\label{sec:proof2}

In this section, we prove Theorem \ref{thm:main2}.

Given a $(3,2$-$2)$ \ssplit \  instance where $\abs{U} = n$
and $\abs{\calS} = m$, we will construct a list of vectors.
Let $\AA \in \mathbb{R}^{m \times n}$ be the incidence matrix 
of the $(3,2$-$2)$ \ssplit \ instance, where $\AA(j,i) = 1$ if element $i \in S_j$
and $\AA(j,i) = 0$ otherwise. Since each set in the $(3,2$-$2)$ \ssplit \ instance has 
$4$ distinct elements, each row of $\AA$ has sum $4$.
Then, we define 
\[
  \MM \defeq \begin{pmatrix}
    \AA & -2\II & -2\II \\
    {\bf 0} & \II - \frac{1}{m}\JJ & {\bf 0} \\
    {\bf 0} & {\bf 0} & \II - \frac{1}{m}\JJ
  \end{pmatrix},
\]
Let $\PPi \defeq \II - \frac{1}{m}\JJ$. 
$\PPi$ is an orthogonal projection matrix onto the subspace orthogonal to $\one_m$.
The dimensions of $\MM$ are $3 m \times (n+2m)$.
Let $D = 3m$ and let $N = n+2m$.
We define a list of vectors $\mathcal{V} = \vv_1, \ldots, \vv_N \in \mathbb{R}^D$ to be the columns
of $\MM$.
Note that each $\vv_i$ has $O(1)$ Euclidean norm.
Dividing each $\vv_i$ by the maximum norm among all $\vv_1, \ldots, \vv_N$ yields 
a list of vectors with Euclidean norm at most $1$.
Without loss of generality, we assume $p \ge 0$.
We define 
\[
\xx_0 \defeq \begin{pmatrix}
  p \one_{n} \\
  (p + (1-p)q) \one_m \\
  (p - (1-p)q) \one_m
\end{pmatrix}.
\]
Denote $\beta = (1-p)q$.
For any $\xx \in \mathbb{R}^N$, 
\[
\sum_{i=1}^N \xx(i) \vv_i  = \MM \xx.  
\]

\begin{claim}
  If $\xx \in \mathbb{R}^N$ satisfies $\mathbb{E}[\xx] = \xx_0$, 
  then $\cov(\MM\xx) = \mathbb{E} \left[ \MM\xx \xx^\top \MM^\top \right]$.
  \label{clm:mx0}
\end{claim}
\begin{proof}
  Note that 
  \[
  \cov(\MM \xx) = \mathbb{E} \left[
    \MM (\xx - \xx_0) (\xx - \xx_0)^\top \MM^\top
  \right].  
  \]
  It suffices to show that $\MM \xx_0 = {\bf 0}$.
  \begin{align*}
    \MM \xx_0 = \begin{pmatrix}
      p \AA \one_n  -2 (p + \beta) \one_m - 2(p - \beta) \one_m \\
      (p + \beta) \PPi \one_m \\
      (p - \beta) \PPi \one_m
    \end{pmatrix}.
  \end{align*}
  Since $\AA \one = 4 \one$ and $\PPi \one = {\bf 0}$,
  we have $\MM \xx_0 = {\bf 0}$.
\end{proof}

Suppose the $(3,2$-$2)$ \ssplit \  instance is satisfiable, and let 
$\zz \in \{\pm 1\}^n$ be an assignment that splits all the sets.
We will show $C(\mathcal{V}, \xx_0) = 0$, that is, there exists a random $\xx \in \mathbb{R}^{N}$ such that 
$\mathbb{E}[\xx] = \xx_0$ and $\cov(\MM\xx) = {\bf 0}$.
 We let 
\[
\xx = \left\{ \begin{array}{ll}
  \one,  & \text{ with probability } p \\
  \begin{pmatrix}
    \zz \\
    \one_m \\
    -\one_m
  \end{pmatrix}, & \text{ with probability } \frac{(1-p)(1+q)}{4} \\
  \begin{pmatrix}
    -\zz \\
    \one_m \\
    -\one_m
  \end{pmatrix}, & \text{ with probability } \frac{(1-p)(1+q)}{4} \\
  \begin{pmatrix}
    -\zz \\
    -\one_m \\
    \one_m
  \end{pmatrix}, & \text{ with probability } \frac{(1-p)(1-q)}{4} \\
  \begin{pmatrix}
    \zz \\
    -\one_m \\
    \one_m
  \end{pmatrix}, & \text{ with probability } \frac{(1-p)(1-q)}{4} 
\end{array}
\right.  
\]

\begin{claim}
  $\mathbb{E}[\xx] = \xx_0$.
\end{claim}

\begin{proof}
  By our setting of $\xx$:
  \begin{align*}
    \mathbb{E}[\xx] & = p \one + 
    \frac{(1-p)(1+q)}{2} \begin{pmatrix}
      {\bf 0} \\
      \one \\
      - \one
    \end{pmatrix}
    + \frac{(1-p)(1-q)}{2} \begin{pmatrix}
      {\bf 0} \\
      -\one \\
      \one
    \end{pmatrix} 
   = p\one + (1-p)q \begin{pmatrix}
    {\bf 0} \\
    \one \\
    - \one
  \end{pmatrix}  = \xx_0.
  \end{align*}
\end{proof}

\begin{claim}
  $\cov(\MM \xx) = {\bf 0}$.
\end{claim}

\begin{proof}
By Claim \ref{clm:mx0}, 
\[
  \cov(\MM\xx) = \mathbb{E}\left[
    \MM \xx \xx^\top \MM^\top
  \right].  
\]
We will show that $\MM\xx = {\bf 0}$ always holds.
We check all the vectors in the support of $\xx$.
If $\xx = \one$, then 
\[
\MM\xx = \begin{pmatrix}
  \AA \one - 4 \one \\
  \PPi \one \\
  \PPi \one
\end{pmatrix}  = {\bf 0}.
\]
If $\xx = \begin{pmatrix}
  \pm \zz \\
  \one \\
  -\one
\end{pmatrix}$, then 
\[
  \MM\xx = \begin{pmatrix}
    \pm \AA \zz - 2\one + 2\one \\
    \PPi \one \\
    - \PPi \one
  \end{pmatrix}  = {\bf 0},
\]
where we use the fact $\AA \zz = {\bf 0}$.
Similarly, if $\xx = \begin{pmatrix}
  \pm \zz \\
  -\one \\
  \one
\end{pmatrix}$, then $\MM\xx = {\bf 0}$.
\end{proof}

Suppose the $(3,2$-$2)$ \ssplit \  instance is $\gamma$-unsatisfiable,
that is, for any $\zz \in \{\pm 1\}^n$, at least $\gamma$ fraction of 
the entries of $\AA \zz$ are in $\{\pm 2, \pm 4\}$.
We will show $C(\mathcal{V}, \xx_0) = \Omega(\beta^2)$, that is, for any random $\xx \in \{\pm 1\}^n$ satisfying $\mathbb{E}[\xx] = \xx_0$,
the operator norm of $\cov(\MM\xx)$ is $\Omega(\beta^2)$.
We write $\xx$ as 
\[
\xx = \begin{pmatrix}
  \xx_1 \\
  \xx_2 \\
  \xx_3
\end{pmatrix},
\]
where $\xx_1 \in \{\pm 1\}^n$ and $\xx_2, \xx_3 \in \{\pm 1\}^m$.
Then,
\begin{align*}
  \MM \xx = \begin{pmatrix}
    \AA \xx_1 - 2 (\xx_2 + \xx_3) \\
    \PPi \xx_2 \\
    \PPi \xx_3
  \end{pmatrix}.
\end{align*}
The following claim splits $\norm{\cov(\MM\xx)}$ into three terms.

\begin{claim}
  $\norm{\cov(\MM\xx)}
  \ge \frac{1}{D} \max \left\{ 
    \mathbb{E} \left[\norm{\AA \xx_1 - 2 (\xx_2 + \xx_3)}^2 \right],
    \mathbb{E} \left[ \norm{\PPi \xx_2}^2 \right],
    \mathbb{E} \left[ \norm{\PPi \xx_3}^2 \right]
    \right\}$.
\label{clm:cov}
\end{claim}

\begin{proof}
By Claim \ref{clm:mx0},
\begin{align*}
  \norm{\cov(\MM\xx)} & = \norm{\mathbb{E} \left[ \MM \xx \xx^\top \MM^\top \right]} \\
  & \ge \frac{1}{D} \tr \left(
    \mathbb{E}  \left[ \MM \xx \xx^\top \MM^\top \right]
  \right) \\ 
  & = \frac{1}{D} \mathbb{E} \left[
    \tr \left(
      \MM \xx \xx^\top \MM^\top
    \right)
  \right] 
  \\
  & = \frac{1}{D} \mathbb{E} \left[
      \norm{\MM \xx }^2
  \right] \\
  & = \frac{1}{D} \left(
    \mathbb{E} \left[ \norm{\AA \xx_1 - 2 (\xx_2 + \xx_3)}^2 \right]
    + \mathbb{E} \left[ \norm{\PPi \xx_2}^2 \right]
    + \mathbb{E} \left[ \norm{\PPi \xx_3}^2 \right]
  \right) \\
  & \ge \frac{1}{D} \max \left\{ 
    \mathbb{E} \left[\norm{\AA \xx_1 - 2 (\xx_2 + \xx_3)}^2 \right],
    \mathbb{E} \left[ \norm{\PPi \xx_2}^2 \right],
    \mathbb{E} \left[ \norm{\PPi \xx_3}^2 \right]
    \right\}.
\end{align*}
\end{proof}

We will show that at least one of the three terms in the rightmost-hand side is sufficiently large.

We first look at the last two terms $\norm{\PPi \xx_2}^2$ and $\norm{\PPi \xx_3}^2$.
Let $\yy \in \{\pm 1\}^m$ be any vector. Then, 
\begin{align*}
  \norm{\PPi \yy}^2
  = \norm{\yy - \frac{\yy^\top \one}{m} \cdot \one}^2.
\end{align*}
Let $\alpha(\yy) = \frac{\yy^\top \one}{m}$. 
Then, $\norm{\PPi \yy}^2 = (1- \alpha(\yy)^2) m$.

\begin{claim}
  If $\mathbb{E}[\alpha(\xx_2)^2]$ or $\mathbb{E}[\alpha(\xx_3)^2]$
  smaller than $1 - \frac{\gamma \beta^2}{24}$, then 
  $\cov(\MM\xx) = \Omega(\gamma \beta^2)$.
  \label{clm:small_alpha}
\end{claim}

\begin{proof}
  Without loss of generality, we assume $\mathbb{E}[\alpha(\xx_2)^2]
  < 1 - \frac{\gamma \beta^2}{24}$. The argument for $\mathbb{E}[\alpha(\xx_3)^2]$
  is the same.
  By Claim \ref{clm:cov}, 
  \begin{align*}
    \norm{\cov(\MM\xx)} 
    \ge \frac{1}{D} \mathbb{E} \left[ \norm{\PPi \xx_2}^2 \right]
    = \frac{(1-\mathbb{E}[\alpha(\yy)^2]) m}{D} 
    = \Omega(\gamma \beta^2).
  \end{align*}
\end{proof}

In the rest of the proof, we assume that 
both $\mathbb{E}[\alpha(\xx_2)^2]$ and $\mathbb{E}[\alpha(\xx_2)^2]$ 
are at least $1 - \frac{\gamma \beta^2}{24}$.
We will show that under this assumption, with probability $\Omega(\beta)$, a large fraction of the
entries of $\xx_2 + \xx_3$ are $0$, and thus $\norm{\AA \xx_1 - 2(\xx_2 + \xx_3)}^2 \approx \norm{\AA \xx_1}^2$. 
This implies $\mathbb{E} \left[ \norm{\AA \xx_1 - 2(\xx_2 + \xx_3)}^2 \right] = \Omega(\beta m)$.
We will need the following properties about $\alpha(\yy)$ for any vector $\yy \in \{\pm 1\}^m$.
When the context is clear, we use $\alpha$ for $\alpha(\yy)$.

\begin{claim}
  For any $\delta > 0$, 
  $\Pr(\abs{\alpha} \le \delta) \le \frac{1 - \mathbb{E}[\alpha^2]}{1 - \delta^2}$.
  \label{clm:alpha}
\end{claim}
\begin{proof}
  Note that 
  \begin{align*}
    \mathbb{E}[\alpha^2] \le \Pr(\abs{\alpha} \le \delta) \cdot \delta^2 + 
    1 - \Pr(\abs{\alpha} \le \delta).
  \end{align*}
  Thus,
  \begin{align*}
    \Pr(\abs{\alpha} \le \delta)
    \le \frac{1 - \mathbb{E}[\alpha^2]}{1 - \delta^2}.
  \end{align*}
\end{proof}

\begin{claim}
  For any $\delta > 0$,
  \begin{align*}
    & \Pr(\alpha > \delta) \ge \frac{\mathbb{E}[\alpha] + \delta(1-2\Pr(\abs{\alpha} \le \delta))}{1+\delta}, \\
    &  \Pr(\alpha < -\delta) \ge \frac{ - \mathbb{E}[\alpha] + \delta(1-2\Pr(\abs{\alpha} \le \delta))}{1+\delta}.
  \end{align*}
  \label{clm:alpha1}
\end{claim}

Remark. Since $-1 \le \mathbb{E}[\alpha] \le 1$, the above two lower bounds are both smaller than $1$.

\begin{proof}
  We introduce some notations:
  \[
  \pi = \Pr(\abs{\alpha} \le \delta), 
  ~ \pi_+ = \Pr(\alpha > \delta),
  ~ \pi_- = \Pr(\alpha < -\delta).
  \]
  Let $\mathbbm{1}_{\abs{\alpha} > \delta}$ be the indicator 
  such that $\mathbbm{1}_{\abs{\alpha} > \delta} = 1$ if $\abs{\alpha} > \delta$
  and $\mathbbm{1}_{\abs{\alpha} > \delta} = 0$ otherwise.
  Then,
  \begin{align*}
    \mathbb{E}\left[ \alpha \mathbbm{1}_{\abs{\alpha} > \delta} \right]
    = \mathbb{E}[\alpha] - \mathbb{E}\left[ \alpha \mathbbm{1}_{\abs{\alpha} \le \delta} \right]
    \ge \mathbb{E}[\alpha] - \delta \pi.
  \end{align*}
  On the other hand, 
  \begin{align*}
    \mathbb{E}\left[ \alpha \mathbbm{1}_{\abs{\alpha} > \delta} \right] 
    \le \pi_+ - \delta \pi_-
    = \pi_+ - \delta (1- \pi_+ - \pi)
    = (1+\delta) \pi_+ + \delta \pi - \delta.
  \end{align*}
  Combining the above two inequalities:
  \begin{align*}
    \pi_+ & \ge \frac{\mathbb{E}[\alpha] + \delta(1-2\pi)}{1+\delta}.
  \end{align*}
  To lower bound $\pi_-$, we note that 
  \begin{align*}
    \mathbb{E} \left[
      \alpha \mathbbm{1}_{\abs{\alpha} > \delta} \right]
      \le \mathbb{E}[\alpha] + \delta \pi.
  \end{align*}
  On the other hand,
  \begin{align*}
    \mathbb{E} \left[
      \alpha \mathbbm{1}_{\abs{\alpha} > \delta} \right]
      \ge \delta \pi_+ - \pi_-
      = \delta (1-\pi - \pi_-) - \pi_-
      = -(1+\delta) \pi_- - \delta\pi + \delta.
  \end{align*}
  Combining the above two inequalities:
  \begin{align*}
    \pi_- \ge \frac{ - \mathbb{E}[\alpha] + \delta(1-2\pi)}{1+\delta}.
  \end{align*}
\end{proof}

We choose $\delta = 1 - \frac{\gamma \beta}{10}$. Here, we choose this number 
to simplify our calculation; we can choose $\delta$ to be any number such that 
$\gamma - (1-\delta)$ is greater than a positive constant.
By our choice for $\xx_0$,
\begin{align*}
  \mathbb{E}[\alpha(\xx_2)] = p + \beta 
  \text{ and }
  \mathbb{E}[\alpha(\xx_3)] = p - \beta.
\end{align*}
Let $\calE$ be the event that both $\alpha(\xx_2) > \delta
\text{ and } \alpha(\xx_3) < -\delta$ happen.
By union bound,
\begin{align*} \Pr \left(
    \calE
  \right) 
  \ge & \Pr\left( \alpha(\xx_2) > \delta \right)
  + \Pr \left( \alpha(\xx_3) < -\delta
  \right)  - 1  \\
  \ge & \frac{\mathbb{E}[\alpha(\xx_2)] - \mathbb{E}[\alpha(\xx_3)] + 2 \delta - 2 \delta \left(\Pr(\abs{\alpha(\xx_2)} \le \delta) + \Pr(\abs{\alpha(\xx_3)} \le \delta) \right)}{1 + \delta} - 1 
  \tag{by Claim \ref{clm:alpha1}} \\
  \ge & \frac{2\beta + 2 \delta - \frac{2\delta}{1-\delta^2} (2 - \mathbb{E}[\alpha(\xx_2)^2] - \mathbb{E}[\alpha(\xx_3)^2]) }{1 + \delta} - 1  
  \tag{by Claim \ref{clm:alpha}} \\
  \ge & \frac{1}{1+\delta} \left(
    2 \beta + 1 - \frac{\gamma\beta}{10} - 4 \left( 1 - \frac{\gamma\beta}{10} \right) \left(\frac{ \gamma \beta^2 / 24}{1-\left(1 - \gamma \beta / 10\right)^2}\right) - 1
  \right) \tag{by our setting of $\delta$ and assumption on $\mathbb{E}[\alpha(\xx_2)^2], \mathbb{E}[\alpha(\xx_3)^2]$} \\
  \ge & \frac{1}{1+\delta} \left(
    2 \beta  - \frac{\gamma\beta}{10} -  \left( 1 - \frac{\gamma\beta}{10} \right) \beta 
  \right) \\
  \ge & \frac{\beta}{2} \left(1 - \frac{\gamma}{10}\right). 
  \tag{by $1+\delta \le 2$}
\end{align*}
Assume event $\calE$ happens. 
At least 
$$\frac{1+\alpha(\xx_2)}{2} > 1 - \frac{\gamma \beta}{24}$$
fraction of the entries of $\xx_2$ are $1$; 
at least 
\[
\frac{1-\alpha(\xx_3)}{2} > 1 - \frac{\gamma \beta}{24}  
\]
fraction of the entries of $\xx_3$ are $-1$.
Thus, at least $1 - \frac{\gamma \beta}{12}$ fraction of the entries 
of $\xx_2 + \xx_3$ are $0$.
Note that among these $0$-valued entries of $\xx_2 + \xx_3$, at least 
$\gamma(1-\frac{\beta}{12})$ fraction of the entries of $\AA \xx_1$ are in $\{\pm 2, \pm 4\}$.
In this case,
\begin{align*}
  \norm{\MM \xx}^2
  \ge \norm{\AA \xx_1 - 2(\xx_2 + \xx_3)}^2
  \ge 4 \gamma \left(1 - \frac{\beta}{12}\right) m.
\end{align*}
By Claim \ref{clm:cov},
\begin{align*}
  \norm{\cov(\MM\xx)}
  & \ge \frac{1}{D} \mathbb{E} \left[
    \norm{\AA \xx_1 - 2(\xx_2 + \xx_3)}^2
  \right] \\
  & \ge \frac{1}{D} \Pr(\calE) \cdot 4 \gamma \left(1 - \frac{\beta}{12}\right) m \\
  & \ge \frac{1}{5m} 
  \cdot \frac{\beta}{2} \left(
    1 - \frac{\gamma}{10}
  \right) \cdot 4 \gamma \left(1 - \frac{\beta}{12}\right) m \\
  & = \Omega(\beta).
\end{align*}
Together with Claim \ref{clm:small_alpha}, a $\gamma$-unsatisfiable $(3,2$-$2)$ \ssplit \ instance 
leads to $C(\mathcal{V}, \xx_0) = \Omega(\beta^2)$.
If we can distinguish whether $C(\mathcal{V}, \xx_0) = 0$ or $C(\mathcal{V}, \xx_0) = \Omega(\beta^2)$, 
then we can distinguish whether a $(3,2$-$2)$ \ssplit \ instance is  satisfiable
or $\gamma$-unsatisfiable, which is NP-hard by Theorem \ref{thm:322sspit_hard}.
This completes the proof of Theorem \ref{thm:main2}.

\bibliographystyle{alpha}
\bibliography{ref}

\end{document}